\documentclass[nocccite,noprefixes,globalnumbering]{cc}

\usepackage[numbers]{natbib}
\usepackage{url}
\usepackage[bookmarks]{hyperref}
\usepackage{paralist}
\usepackage{float}
  \usepackage{graphicx}
  \usepackage[labelformat=empty, labelsep=colon]{caption}
  \restylefloat{figure}

\newcommand{\size}[1]{{\left|#1\right|}}
\newcommand{\pair}[1]{{\left\langle#1\right\rangle}}
\DeclareMathOperator{\poly}{poly}

\contact{teutsch@cse.psu.edu}

\submitted{3 January 2013}

\title{Short lists for shortest descriptions in short time}

\author{
Jason Teutsch \\
\emph{Penn State University}
}

\begin{abstract}
Is it possible to find a shortest description for a binary string? The well-known answer is ``no, Kolmogorov complexity is not computable.''  Faced with this barrier, one might instead seek a short list of candidates which includes a laconic description.  Remarkably such approximations exist.  This paper presents an efficient algorithm which generates a polynomial-size list containing an optimal description for a given input string.  Along the way, we employ expander graphs and randomness dispersers to obtain an Explicit Online Matching Theorem for bipartite graphs and a refinement of Muchnik's Conditional Complexity Theorem.  Our main result extends recent work by Bauwens, Mahklin, Vereschchagin, and Zimand.
\end{abstract}


\begin{subject}
68Q30, 68R10
\end{subject}

\begin{document}

\maketitle

\section{The quest for short descriptions}

We explore an interaction between randomness extraction, combinatorics, and Kolmogorov complexity culminating in an efficient, new approximation for optimal descriptions.  Informally, a computer program $p$ is called a description for a binary string $x$ if the execution of $p$ yields output $x$. The Kolmogorov complexity of a binary string is the length of its shortest description in some standard programming language (see Section~\ref{sec: conventions}).  As much as one might like to know the Kolmogorov complexity of a given string, it is impossible to obtain this quantity effectively \cite{LV08}.  Even estimating Kolmogorov complexity for a given string is infeasible, as no unbounded computable function can be a lower bound for a Kolmogorov complexity \cite[Theorem~1.6]{ZL70}.  Moreover, any algorithm mapping a string to a list of values containing its Kolmogorov complexity must, for all but finitely many lengths~$n$, include in the list for some string of length~$n$ at least a fixed fraction of the lengths below $n+O(1)$~\cite{BBFFGLMST06}.

Remarkably, as recently observed by Bauwens, Makhlin, Vereshchagin, and Zimand \cite{BMVZ12}, the situation differs when we seek a short list of candidate descriptions for a given string.  We will show that it is possible to efficiently compute a polynomial-size list containing a shortest description for any given string, up to an additive constant number of bits (Corollary~\ref{cor: jason}).  The existence of our listing algorithm will follow from a combinatorial graph construction, namely our Explicit Online Matching Theorem, and we devote the remainder of our discussion to establishing this crucial result.

Sections~\ref{sec: conventions} and~\ref{sec: parr} provide background and prior results.  Section~\ref{sec: randomness extraction} discusses the bipartite expander and disperser graphs which we will use to obtain our main theorems in Section~\ref{sec: eom}, and the final section provides additional analysis of the core construction.

\section{Conventions for complexity and bipartite graphs} \label{sec: conventions}
We formalize the notions of ``description'' and ``Kolmogorov complexity'' from the previous section, review the definition of conditional complexity, and then discuss bipartite graphs.  Throughout this manuscript, $\size{x}$ denotes the length of a string $x$, and $\size{S}$ denotes the cardinality of a set $S$.  

For every Turing machine~$M$, we call
\[
C_M(x) = \min \{\size{p} \colon M(p) = x\}
\]
the (plain) \emph{Kolmogorov complexity with respect to~$M$}.  Let \linebreak $M_0, M_1, \dotsc$ be an effective enumeration of all Turing machines.  Let $\pair{\cdot, \cdot}$ denote a polynomial-time computable encoding for pairs of strings whose output has length which is a linear function of the first coordinate's length plus the second coordinate's length, and define a machine $U$ by $U(\pair{e,x}) = M_e(x)$.  This \emph{standard} machine~$U$ has the property that for any further machine $M$, there exists a constant~$d$ such that $C_U(x) \leq C_M(x) + d$ for all strings~$x$, see \cite{LV08} for details.  For the remainder of this manuscript, let $C = C_U$ denote the Kolmogorov complexity of the standard machine~$U$.  We will say that $p$ is a \emph{description} of a binary string $x$ if $U(p) = x$.  When discussing pairs, we may omit the delimiters $\pair{\cdot}$ for readability.

\emph{The conditional complexity of $a$ given $b$}, or $C(a \mid b)$, is the length of the shortest string which translates the string $b$ into the string $a$.  More specifically,
\[
C(a \mid b) = \min \{ \size{p} \colon U(p,b) = a \}.
\]

We will use the following notation for graphs.  Triplets $(L,R,E)$ will denote bipartite graphs in which $L$ is a set of left-hand vertices, $R$ is a set of right-hand vertices, and $E \subseteq L \times R$ is a set of edges connecting these two halves.  For any set of vertices $S$, $E(S)$ denotes the neighbors of $S$, and a bipartite graph has \emph{left degree} $d$ if each of its left-hand vertices has exactly $d$ neighbors in $R$.  

A family of bipartite graphs is called \emph{explicit} if for each member $(L,R,E)$ with left-degree~$d$, the $i^\text{th}$ neighbor of any vertex in~$L$ can be computed in time polynomial in $\pair{\log \size{L}, \log d}$.  When the family context is clear, we simply say that $(L,R,E)$ is itself \emph{explicit}.  Similarly, an infinite bipartite graph whose left-hand vertices consist of binary strings is called \emph{explicit} if the $i^\text{th}$ neighbor of each left-hand vertex $x$ can be computed in time $\poly (\size{x}, \log \size{E(\{x\})})$.

\section{Prior and related results} \label{sec: parr}

Buhrman, Fortnow, Laplante \cite{BFL01} and Muchnik \cite{Muc02} observed that if a hash value determines a string more or less uniquely among a class of strings, then that hash value serves as a description of that string modulo advice.  While the authors of \cite{BFL01} used randomness extraction to obtain bounds\footnote{Bauwens, Makhlin, Vereshchagin, and Zimand  \cite{BMVZ12} have refined the distinguishing complexity result from \cite{BFL01}.} on polynomial-time variants of Kolmogorov complexity, \cite{Muc02} employed probabilistic methods to prove the existence of an expander-like object.  The latter translates into the next statement about conditional complexity.
\begin{namedtheorem*}{Muchnik's Conditional Complexity Theorem}[\citep{Muc02}]
For any strings $x$ and $y$, there exists a string $p$ such that
\begin{enumerate}[\scshape (i)]
\item $C(x \mid p,y) = O(\log \size{x})$,
\item $\size{p} = C(x \mid y)$, and 
\item $C(p \mid x) = O(\log \size{x})$.
\end{enumerate}
The hidden constants do not depend on $x$, $y$, or $p$.
\end{namedtheorem*}

\pagebreak[2]
Muchnik's Theorem has a simple interpretation in the context of multisource information theory \cite{MRS11}.  Suppose Alice wants to send a string $x$ to Bob and that Bob already knows $y$. By definition Alice must send a message $p$ of length at least $C(x \mid y)$ bits to Bob in order to communicate $x$.  Muchnik's Theorem tells us that Alice can construct such a message $p$ without even knowing $y$!   According to the theorem, Alice requires logarithmic advice to encode the message $p$, and Bob then needs just logarithmic advice to transform $p$ back into $x$.

A recent paper \cite{MRS11} by Musatov, Romashchenko, and Shen furnished two combinatorial proofs of Muchnik's Conditional Complexity Theorem.  The first proof, which introduced online matchings (Definition~\ref{defn: olm}), roughly follows Muchnik's original argument \cite{Muc02} whereas the second one appeals to randomness extraction along the lines of Buhrman, Fortnow, and Laplante \cite{BFL01}.  The proof of our main result, Corollary~\ref{cor: jason}, combines both of these methods in its core construction.  Somewhat unexpectedly, we can also understand Muchnik's Conditional Complexity Theorem in terms of lists of descriptions, the main objects of this paper.  We explore the details of this connection in Corollary~\ref{cor: bettermuchnik}.

Building on ideas from \cite{MRS11}, Bauwens, Makhlin, Vereshchagin, and Zimand \cite{BMVZ12} improved the following theorem in two ways. 
\begin{theorem}[Bauwens, Makhlin, Vereshchagin, and Zimand {[\citep{BMVZ12}]}] \label{thm: easylog}
There exists a computable function which maps each binary string $x$ to a $\poly(\size{x})$-size list containing a length $C(x) + O(\log \size{x})$ description for $x$.
\end{theorem}
They showed that either one can generate the list in polynomial time, or one can bound the length of the contained description by $C(x) + O(1)$.  In the latter case the length of the list can be quadratic in $\size{x}$, and furthermore no computable function can generate a shorter list for descriptions of this size \cite{BMVZ12}.  The authors also give an improvement of Muchnik's Conditional Complexity Theorem.  They show that the description $p$ in \textsc{(iii)} can be computed \emph{efficiently} from $x$ given $O(\log \size{x})$ bits of advice, and we will improve their result further in Corollary~\ref{cor: bettermuchnik} by showing that the number of advice bits in \textsc{(i)} can be reduced from $O(\log \size{x})$ to $O(1)$.  In a similar vein, Musatov and Romashchenko \cite{Mus11, MRS11} investigated Muchnik's Theorem in the context of space complexity.

In this paper, we shall show that both the Theorem~\ref{thm: easylog} improvements from \cite{BMVZ12} can be achieved simultaneously.  We will efficiently generate a polynomial-length list containing a description whose length is within an additive constant of optimal (Corollary~\ref{cor: jason}).  We do not investigate running time of descriptions here, however time complexity for these objects remains a relevant consideration.

\section{Randomness extraction tools} \label{sec: randomness extraction}

Our main construction combines disperser and expander graphs which we derive from an explicit graph object of Ta-Shma, Umans, and Zuckerman~\cite{TUZ07}.

\begin{definition}
A bipartite graph $(L,R,E)$ is called a $(K,\epsilon)$-\emph{disperser} if every subset of $L$ with cardinality at least $K$ has at least $(1-\epsilon)\size{R}$ distinct neighbors in $R$.
\end{definition}

\begin{definition} \label{defn: expander}
A bipartite graph $(L,R,E)$ is called a $(K,c)$-\emph{expander} if for every set $S \subseteq L$ of size at most $K$, $\size{E(S)} \geq c\size{S}$.
\end{definition}

Both expanders and dispersers are bipartite graphs whose left-hand vertices have many neighbors on the right side, however they differ in two important respects.  First, expander graphs achieve expansion only for sufficiently small sets whereas disperser graphs only guarantee dispersion for large sets.  Secondly, while both expanders and dispersers involve graphs with small left degree whose left-hand subsets have many neighbors, they differ in their means of achieving these parameters.  In a bipartite expander graph, subsets of left-hand vertices have many right-hand neighbors relative to their size, whereas in a disperser graph the neighbors of the left-hand subset cover a large fraction of the entire right-hand side.

\begin{theorem}[Ta-Shma, Umans, Zuckerman {[\citep{TUZ07}]}] \label{thm: tuz}
There is a constant $c>0$ such that for every $K \geq 0$, $\epsilon > 0$, and any nonempty set of vertices $L$, there exists an explicit $(K,\epsilon)$-disperser $(L,R,E)$ with left degree $d=\poly \log\size{L}$ and 
\[
\frac{cKd}{\log ^3\size{L}} \leq \size{R} \leq Kd.
\]
\end{theorem}

Ta-Shma, Umans, and Zuckerman's construction aims to make the set of right-hand vertices as large as possible, however our purposes require a small right-hand set.

\begin{namedtheorem*}{Disperser Lemma}[Zimand] \label{lem: marius}
For every~$k \geq 0$ and any set~$L$ with $\size{L} \geq 2^k$, there is an explicit bipartite graph $(L,R,E)$
\begin{itemize}
\item with $\size{R} = 2^{k+1}$,
\item whose left degree is polynomial in $\log \size{L}$ and does not depend on $k$, and such that
\end{itemize}
any subset of $L$ of size at least $2^k$ has at least $2^k$ neighbors in~$R$.
\end{namedtheorem*}

\begin{proof}
Let $k \geq 0$ and $L$ be a set of size at least $2^k$.  Apply Theorem~\ref{thm: tuz} with $L$, $K = 2^k$ and $\epsilon = 1/3$, and call the resulting graph~$G = (L,R,E)$.  $G$ is close to what we need, however the right-hand size vertex set could be either too large or too small.  If $ \size{R} = 2^{k+1}$, in which case we say that $R$ has the \emph{right size}, then $G$ is already what we need as any subset of $L$ of size at least $2^k$ has at least $(2/3) 2^{k+1} > 2^k$ neighbors in $R$.

Consider the case where $R$ is too small, meaning less than twice the right size (but not already equal to it).  In this step we will overshoot the size of $R$ by a bit, and then we correct for this in the next paragraph.  We increase the size of the right-hand set by merging cloned copies of $G$.  Form a new graph which has the same left-hand vertices as~$G$, whose right-hand vertices are a disjoint union of $R$ with itself, and whose edges are the same as the ones for $G$ in each half.  This operation doubles both the size of the right-hand vertex set and the degree of the graph while maintaining the disperser parameters $K = 2^k$ and $\epsilon = 1/3$.  We iterate this operation until the right-hand vertex set becomes at least twice the right size.  The resulting graph preserves the disperser parameters, and the degree is still $\poly \log \size{L}$ as Theorem~\ref{thm: tuz} provided us with a graph whose right-hand cardinality was already no less than $O(1/ \log^3 \size{L})$ times the right size.

Without loss of generality, assume that $R$ is at least twice the right size.  We divide $R$ into $2^{k+1}$ equivalence classes of approximately equal size, and call this collection of classes $R'$.  Specifically, we distribute the vertices of $R$ evenly among the classes so that no class is bigger than any other by more than one member.  Now define a bipartite graph $G'$ with left vertex set $L$, right vertex set $R'$, and where $x \in L$ is a neighbor of $y \in R'$ iff $(x,z) \in E$ for some $z$ in the equivalence class of $y$.  We claim that $G'$ has the desired properties for the lemma.  $R'$ is already the right size, and the folding operation just described does not increase the left degree, so it remains to verify the disperser property.  Note that each equivalence class must contain at least two vertices since $R$ is at least twice the right size.  Let $t\geq 2$ be the unique integer so that $t \leq \size{R}/2^{k+1} < t+1$, and let $S$ be a subset of $L$ of size at least $2^k$.  Without loss of generality, we can assume that $\size{E(S)}$ is exactly equal to $(1-\epsilon) \size{R}$ because if $S$ had more neighbors we would get even better parameters for the disperser.  Thus the number of equivalence classes in $R'$ which contain no neighbor of $S$ is greater than $\epsilon \size{R} / (t+1)$ and at most $\epsilon \size{R} / t$.  Indeed
\[
\epsilon 2^{k+1} = \frac{\epsilon \size{R}}{\size{R} / 2^{k+1}}
\]
lies between these two values.  The ratio between these two endpoints is $(t+1)/t \leq 3/2$, and it follows that the number of elements in $R'$ which have no neighbor in $S$ is at most (in fact strictly less than) $(3/2) \epsilon 2^{k+1} = (1/2) 2^{k+1}$, so $S$ has at least $2^{k+1} - 2^k$ neighbors in $R'$.
\end{proof}

We modify the above construction to obtain an expander graph.

\begin{namedtheorem*}{Expander Lemma}
For every~$k \geq 0$ and any set~$L$ with $\size{L} \geq 2^k$, there is an explicit bipartite graph $(L,R,E)$
\begin{itemize}
\item where $\size{R} < 2^{k+3}$,
\item whose left degree is polynomial in $\log \size{L}$ and does not depend on $k$, and such that
\end{itemize}
any subset $S \subseteq L$ of size at most $2^k$ has at least $\size{S}$ neighbors in~$R$.
\end{namedtheorem*}

\begin{proof}
Assume $k \geq 0$, and let $L$ be a set satisfying $\size{L} \geq 2^k$.  First we construct for each $0 \leq i \leq k$ a disperser graph like the one in the previous lemma but with slightly different parameters.  We want a bipartite graph $G_i= (L,R_i,E_i)$ whose left degree is polynomial in $\log \size{L}$, where $\size{R_i}= 2^{i+2}$, and such that any subset of $L$ of size at least $2^i$ has at least $2^{i+1}$ neighbors in $R_i$.  The same argument from the previous lemma gets us a graph with these parameters when we alter the ``right size'' and number of equivalence classes to be $2^{i+2}$.

We now transform this collection of disperser graphs into an expander graph.  Let $G = (L,R,E)$ be a merge of all these disperser graphs.  That is, $R$ is the disjoint union of $R_0, R_1, \dotsc, R_k$, and $E$ is the corresponding union of the $E_i$'s.  The left degree of $G$ is at most $k$ times the maximum left degree of all $G_i$'s, which is $\poly \log \size{L}$, and
\[
\size{R} = \sum_{i=0}^k \size{R_i} = \sum_{i=0}^k 2^{i+2} = 2^{k+3}-4.
\]
Consider any $S \subseteq L$ of size at most $2^k$, and let $i$ be the unique integer such that $2^i \leq \size{S} < 2^{i+1}$.  Then
\[
\size{E(S)} \geq \size{E_i(S)} \geq 2^{i+1} > \size{S}. \qed
\]
\end{proof}

\section{Explicit online matching} \label{sec: eom}

We now present the main theorem.  The core of our constructions is the following ``static'' disperser graph which we transform into a further bipartite graph that admits ``online'' matching.  In case one does not require an explicit graph, a bipartite graph with randomly chosen edges achieves the other properties of Lemma~\ref{lem: fullk7hall} with nonzero probability\footnote{See also \cite{MRS11} for an example of a similar but finite graph constructed using the probabilistic method.} \cite{BMVZ12}. 

\begin{lemma} \label{lem: fullk7hall}
For every $k \geq 0$, there exists an explicit bipartite graph $(L,R,E)$ such that
\begin{itemize}
\item $L$ consist of all binary strings of length at least~$k$,
\item the cardinality of $R$ is $2^{k+1}$,
\item the degree of each vertex $x \in L$, is $\poly(\size{x})$, and
\item where any subset of $L$ of size at least $2^k$ has at least $2^k$ neighbors in $R$.
\end{itemize} 
The polynomial $\poly(\size{x})$ does not depend on~$k$.
\end{lemma}

\begin{proof}
Our construction proceeds in two phases.  First, we use the Expander~Lemma to spread the neighbors of the left-hand nodes $L$ across a small middle vertex set~$M$.  Next we take this spread of neighbors $M$ and map it to an even smaller set, namely the right-hand vertices~$R$, via the Disperser~Lemma.  The edges $E$ of our desired graph will consist of those pairs in $L$ and $R$ which are connected by the composition of these two mappings.

We now discuss how to handle strings of different lengths.  Each string in $L$ of length greater than $2^k$ will have $2^k$ neighbors in $R$ (which is polynomially many).  For each length~$n$ in the remaining range $k \leq n \leq 2^k$, we create an explicit bipartite expander between strings of length~$n$ and an intermediate set $M_n$ and then disperse the disjoint union of the $M_n$'s into the set $R$.

\begin{figure}
\caption{\textbf{Figure.} The Expander~Lemma composed with the Disperser~Lemma.}
\begin{center}
\includegraphics[height=7.25cm]{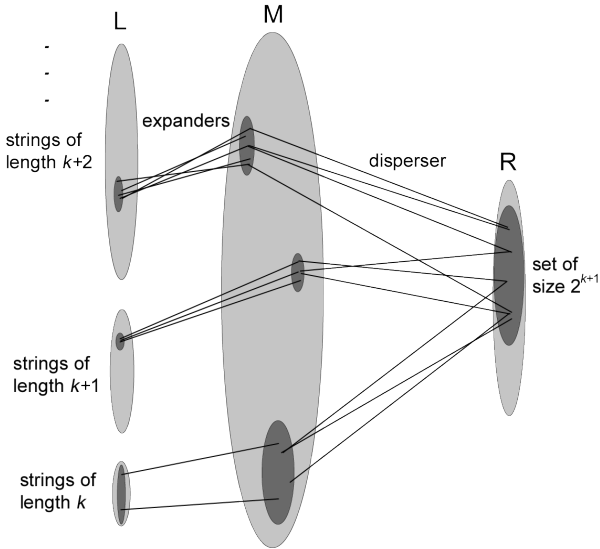}
\end{center}
\end{figure}

Let $L$ be as in the assumption of this lemma, let $k \geq 0$, and let $L_n$ denote the strings of length~$n$.  We generate an explicit expander graph $G_n = (L_n, M_n, A_n)$ for each string length $k \leq n \leq 2^k$. Apply the Expander Lemma to obtain a bipartite graph $G_n$ with left degree $\poly (n)$ and right vertex cardinality $2^{k+3}$ which is a $(2^k, 1)$-expander.  Let $M$ be the disjoint union of $M_k, M_{k+1}, \dotsc, M_{2^k}$, and let $A$ denote the corresponding edges of this embedding.  That is,
\[
A = \{(x,y) \colon \text{$x \in L$, $y \in M$ and $(x,y) \in A_n$ for some $n$}\}.
\]
Now $(L, M, A)$ is a bipartite graph where the left degree of each length~$n$ string is $\poly(n)$, and
\[
2^k < \size{M} = \sum_{n=k}^{2^k} \size{M_n}
< 2^k \cdot 2^{k+3}.
\]

By the Disperser~Lemma, there exists an explicit bipartite graph $(M, R, B)$ whose left degree, $\poly(\log \size{M})$, is polynomial in $k$, whose right vertex set satisfies $\size{R} = 2^{k+1}$, and such that any subset of~$M$ of size at least $2^k$ has at least $2^k$ neighbors in $R$.  Finally, let $E$ consist of all edges between $L$ and $R$ which are connected through $M$ by composition of the edge sets $A$ and $B$.  That is,
\begin{multline*}
E = \{ (x,y) \colon \text{$x \in L$, $y \in R$, and there exists $z \in M$} \\ \text{such that $(x,z) \in A$ and $(z,y) \in B$}\}.
\end{multline*}
This concludes the construction of the bipartite graph $(L,R,E)$, whose left degree for each string of length $n$ is no more than 
\[
\max\{O(n), \poly(n) \cdot \poly(\log \size {M})\} = \poly (n),
\]
regardless of whether or not $n > 2^k$.

Let $S$ be a subset $L$ of size at least $ 2^k$.  If $S$ contains a string of length greater than $2^k$ then that string has $2^k$ neighbors in $R$, so in this case we immediately satisfy $\size{E(S)} \geq 2^k$.  Thus we may assume that all strings in $S$ have length at most $2^k$.  It follows from the expansion property of $(L,M,A)$ that
\[
\size{A(S)} = \sum_{n=k}^{2^k} \size{A_n(\{x \in S \colon \size{x} = n \})} \geq 2^k,
\]
and from the disperser property of $(M,R,B)$ we then get
\[
\size{E(S)} = \size{B[A(S)]} \geq 2^k
\]
as desired.
\end{proof}

The next definition and the argument in the next proof are due to Musatov, Romashchenko, and Shen \cite{MRS11}, but we include them here for clarity and completeness. 

\begin{definition} \label{defn: olm}
We say that a bipartite graph $(L,R,E)$ \emph{admits online matchings up to size $s$} if there exists an algorithm such that for any set of vertices in $L$ of size~$s$, whose vertices are (adversarially) presented to the algorithm one at a time, the algorithm can assign each vertex in order received to one of its neighbors (without knowing what comes next), and the overall assignment after all $\leq s$ elements is a bijection.  An online matching is \emph{efficient} if a neighbor can be selected in time linear in the logarithm of the degree of the input.
\end{definition}

We will use the following combinatorial theorem to obtain our main result about Kolmogorov complexity.

\begin{namedtheorem*}{Explicit Online Matching Theorem}
For every $k\geq 0$, there exists an explicit bipartite graph $G = (L,R,E)$ such that
\begin{itemize}
\item $L$ consist of all binary strings of length at least~$k$,
\item the cardinality of $R$ is less than $2^{k+1}$,
\item the degree of each vertex $x \in L$ is $\poly(\size{x})$, and
\item $G$ admits efficient online matching up to size $2^k$.
\end{itemize}  
The polynomial $\poly(\size{x})$ does not depend on~$k$.
\end{namedtheorem*}

\begin{proof}
Let $L$ be as in the hypothesis, and let $k \geq 0$.  Apply Lemma~\ref{lem: fullk7hall} to obtain, for each integer $0 \leq i < k$, a bipartite graph $(L, R_i, E_i)$ where $R_i$ is a set of vertices with cardinality~$2^{i+1}$, the left degree is polynomial in the string length, and any subset of $L$ of size at least $2^i$ has at least $2^i$ neighbors in $R_i$.  Furthermore, let $(L, R_{-1}, E_{-1})$ be a bipartite graph which has a single right-hand vertex which is a neighbor of each element in $L$.  We build the $R_i$'s pairwise disjoint.

Let $R = \bigcup_{i \geq -1} R_i$ and $E = \bigcup_{i \geq -1} E_i$.  We claim that the explicit bipartite graph $(L,R,E)$ has the required properties to establish the theorem.  Since the degree of each vertex in $x \in L$ is the sum of the degrees for $x$ over all $R_i$'s, we see that the degree of each length~$n$ string in $L$ is at most $(k+1) \cdot \poly(n)$, which is still polynomial in $n$.  Furthermore,
\[
\size{R} = \sum_{i={-1}}^{k-1} \size{R_i} = \sum_{i=0}^k 2^i = 2^{k+1}-1.
\]

It remains to verify the efficient online matching property.  We apply a greedy algorithm.  Originally, all vertices in $R$ are marked as unused.  When a vertex $x \in L$ comes in, we assign it to an arbitrary unused neighbor $R_{k-1}$, if such a neighbor exists.  If not, we attempt to assign $x$ to an unused neighbor in $R_{k-2}$.  If this is not possible, we try for an unused neighbor in $R_{k-3}$, etc.  When (and if) $x$ gets assigned to a particular $y \in R_i$, we mark $y$ as used and wait for the next vertex in $L$ to arrive.  Successive arrivals are handled similarly.

We claim that every $x$ gets assigned through this method, and since each $x$ is assigned to an unused vertex, the resulting matching will be a bijection.  We argue by induction that no more than $2^{i}$ vertices in $L$ may fail assignment at level $R_i$ for any $i \geq 0$.  Suppose this bound were exceeded at level $R_{k-1}$, and let $X \subseteq L$ denote those vertices which failed assignment at this level.  Then $\size{X} > 2^{k-1}$, and so by the disperser property $\size{E_{k-1}(X)} \geq 2^{k-1}$.  Each element of $E_{k-1}(X)$ must be used, otherwise we could have matched an element of $X$ to it.  Therefore the total number of vertices which were either matched or failed assignment at level $R_{k-1}$ exceeds $2^k$, a contradiction.

Thus at most $2^{k-1}$ vertices fail assignment at level $R_{k-1}$, and of these at most half fail assignment at level $R_{k-2}$, and so by induction, at most $2^0$ vertices have failed assignment at all levels $R_i$ for $i \geq 0$.  The remaining vertex in $R_{-1}$ is used to assist with the recalcitrant vertex, if needed.  Hence all vertices are matched.
\end{proof}

\begin{remark}[Makhlin]
The cardinality of $R$ in the Explicit Online Matching Theorem cannot be reduced to $2^k$, as then the online matching property would force each right-hand vertex of length~$n$ to have more than $2^n - 2^k$ neighbors, which would imply that the graph restricted to strings of length~$n$ has more than $2^{k}(2^n - 2^k)$ edges, and therefore some left-hand vertex of length~$n$ has degree greater than $2^k(1 - 2^k / 2^n)$.  It follows that for every~$k$ there exists a string in $L$ of length $n= k+1$ whose degree exceeds $2^{k-1}$.

On the other hand, for any $\delta>0$ we can achieve $\size{R} < 2^{k+\delta}$.  If we modify $\epsilon$ in the Disperser~Lemma to be $(2/3)(1 - 2^{-\delta})$ rather than $\epsilon = 1/3$ and fold into $2^{k+\delta}$ equivalence classes rather than $2^{k+1}$ of them, we obtain a disperser graph with right-hand size $2^{k+\delta}$ (and other parameters the same).  This smaller right-hand size then carries over to Lemma~\ref{lem: fullk7hall} and the Explicit Online Matching Theorem.
\end{remark}

We now formalize the connection between online matching and short lists.

\pagebreak

\begin{corollary} \label{cor: jason}
There exists a polynomial-time computable function which maps each binary string $x$ to a $\poly(\size{x})$-size list containing a length $C(x) + O(1)$ description for $x$.
\end{corollary}

\begin{proof}
Define a (not necessarily polynomial-time) machine $M$ which does the following.  At first, every right-hand vertices in each $k$-parameter Explicit Online Matching Theorem (EOMT) graph is designated as ``unused.''  Dovetail on all programs $p$ for the standard machine $U$, and as each one converges, apply the EOMT with $k=\size{p}$ to match the value $U(p)$, if it has not been matched already, with an unused right-hand vertex $z$.  If $U(p)$ already had a match in the $k$-parameter EOMT graph, then do nothing.  Otherwise set $M(z) = U(p)$, and mark $z$ as ``used.''

The process just described attempts to match no more than $2^k$ strings on the $k$-parameter EOMT graph because there are only $2^k$ many binary strings of length~$k$, and each of these attempts succeeds because the $k$-parameter EOMT graph admits online matchings up to size $2^k$.  Since the $\size{p}$-parameter EOMT graph has less than $2^{\size{p}+1}$ right-hand vertices, we may interpret $z$ as a string of length $\size{p} + 1$.  Thus whenever $p$ is a description for some string $x$, $x$ is eventually matched to a string $z$ of length $\size{p} + 1$ which is a neighbor of $x$ in the $\size{p}$-parameter EOMT graph and satisfies $M(z) = x$.  In particular, this holds when $p$ is a shortest description for $x$.

Now $M = M_e$ for some index~$e$, hence for all strings~$x$ with $C(x) \leq \size{x}$ the efficiently computable set
\begin{multline*}
\{\pair{e,y} : \text{$y$ is a neighbor of $x$}\\ \text{in some EOMT graph with parameter $k \leq \size{x}$}\}
\end{multline*}
contains a description $\pair{e,z}$ for $x$ with $\size{z} \leq C(x) + 1$.  In order to cover the case $C(x) > \size{x}$, we add an extra description to this set, namely $\pair{i,x}$, where $M_i$ is the identity map on all strings.
\end{proof}

One might wonder whether the additive $O(1)$ error term is necessary.  The answer depends on the underlying standard machine, as shown in \cite{BMVZ12}.  While there exists a standard machine for which Corollary~\ref{cor: jason} holds with $O(1)$ equal to zero, there are also standard machines where having $O(1)$ equal to zero forces the list size to become exponential in $\size{x}$.

As a further corollary, we improve Muchnik's Conditional Complexity Theorem not only by making the description $p$ polynomial-time computable but also by reducing the decoding overhead from $O(\log n)$ bits to a constant.  We define \emph{time-bounded conditional complexity} $C^t$ for a string $a$ relative to a string $b$ as follows:
\[
C^t (a \mid b) = \min \{ \size{p} \colon \text{$U(p,b)$ converges to $a$ in at most $t$ steps}\}.
\]
\begin{corollary} \label{cor: bettermuchnik}
For any strings $x$ and $y$, there exists a string $p$ such that
\begin{enumerate}[\scshape (i)]
\item $C(x \mid p,y) = O(1)$,
\item $\size{p} = C (x \mid y)$, and 
\item $C^{\poly(\size{x})}(p \mid x) = O(\log \size{x})$.
\end{enumerate}
The hidden constants do not depend on $x$, $y$, or $p$.
\end{corollary}

\begin{proof}
First, note that the proof of Corollary~\ref{cor: jason} can easily accommodate conditional complexity: if we dovetail on all programs $\pair{p,y}$ rather than on $p$ and perform computations with $\pair{p,y}$ instead of $p$, then the same construction efficiently computes a $\poly(\size{x})$-size list $f(x)$ containing a length $C(x \mid y) + O(1)$ string $q$ which satisfies $U(q,y) = x$.  Let $p$ be the string $q$ with the last $O(1)$ bits removed.  Then $\pair{p,y}$ together with $O(1)$ bits of advice suffice to reconstruct $x$, $\size{p} = C(x \mid y)$, and
\[
C^{\poly(\size{x})}(p \mid x) \leq C^{\poly(\size{x})}(q \mid x) + O(1) = O(\log \size{x})
\]
since $O(\log \size{x})$ bits are enough to distinguish among the members of $f(x)$.
\end{proof}

\section{How big is the polynomial?} \label{sec: analysis}
We estimate\footnote{Zimand~\cite{Zim13} achieves an improved degree bound of $6+\delta$.} the size of the polynomial-size list in Corollary~\ref{cor: jason}.  Our analysis involves some minor modifications to the construction.  Throughout the discussion below, $\delta$ denotes an arbitrary positive constant, and $n$ is shorthand for $\log  \size{L}$.

First we calculate the left-degree of the expander graph in our main construction, Lemma~\ref{lem: fullk7hall}.  In order to do this, we must first determine the left degree of the Ta-Shma, Umans, Zuckerman disperser in Theorem~\ref{thm: tuz}.  The authors of \cite{TUZ07} state the left degree as $\poly (n)$, but in fact it need not exceed $O(n^3)$.  This sharper bound follows by redoing the composition construction in \cite[Lemma 6.4]{TUZ07}, with the extractor in \cite[Lemma 4.21]{GUV09}.  In more detail, we import the parameters from \cite[Lemma 4.21]{GUV09} as follows:
\begin{align*}
t(n,k) &= \log n + O[\log k \cdot \log (k/\epsilon)], &\text{and \ }
\Delta = 2 \log (1/\epsilon) - O(1).
\end{align*}
The definitions of $E_1$ and $E_2$ in \cite[Lemma 6.4]{TUZ07} stay the same relative to these changes, as does the calculation of the disperser~$E$.  Feeding this new $E$ into \cite[Lemma 6.5]{TUZ07} yields a left degree of $O(n^{2+\delta})$ in Theorem~\ref{thm: tuz}.  This in turn, gives an $O(n^3)$ left degree for the graph in the Disperser~Lemma as the left degree may increase slightly in the case where $R$ has cardinality less than the right size.

By similar inspection, one obtains a bound of $O(n^4)$ for the left degree of the graph in the Expander~Lemma, however we can reduce this to $O(n^{2+\delta})$ by replacing that lemma with the following alternative expander construction.
\begin{theorem}[Guruswami, Umans, Vadhan \citep{GUV09}] \label{thm: guv09}
There exists a constant $c$ such that for every  $\alpha, \epsilon > 0$ and every $0 \leq k \leq n$, there is an explicit bipartite graph $(L,R,E)$ with
\begin{itemize}
\item $\size{L} = 2^n$,

\item left degree $d = c(nk/\epsilon)^{1+ 1/\alpha}$, and

\item  $\size{R} \leq d^2 \cdot 2^{k(1+\alpha)}$ 
\end{itemize} 
which is a $[2^k, (1-\epsilon)d]$-expander.
\end{theorem}
The Guruswami, Umans, and Vadhan expander graph achieves nearly perfect expansion into a space which is small enough to preserve the other parameters of Lemma~\ref{lem: fullk7hall}.  Since the neighbors of each left-hand vertex $x$ in Lemma~\ref{lem: fullk7hall} now derive from the expander graph of Theorem~\ref{thm: guv09} composed with the revised Disperser~Lemma graph from this section, we calculate the degree of $x$ to be $O(\size{x}^{5 + \delta})$.

At this point the cloning procedure in the Explicit Online Matching Theorem contributes a linear term to the $O(\size{x}^{5 + \delta})$, and the union over sets of size $k \leq \size{x}$ in Corollary~\ref{cor: jason} brings the total list size to $O(\size{x}^{7 + \delta})$.  This completes our proof sketch of the following result.

\begin{theorem}
There exists a constant $c \geq 1$ such that for every $\delta > 0$, there exists a nonnegative~$d$ and a polynomial-time computable function which maps each binary string~$x$ to a size $c \cdot \size{x}^{7+\delta}$ list containing a length $C(x) + d$ description for $x$.
\end{theorem}

The list size bound also refines the third part of Corollary~\ref{cor: bettermuchnik}.

\begin{theorem}
There exists a constant $c \geq 0$ such that for every $\delta >0$, there exists a nonnegative~$d$ such that for any strings $x$ and $y$, there exists a string $p$ such that
\begin{enumerate}[\scshape (i)]
\item $C(x \mid p,y) = d$,
\item $\size{p} = C (x \mid y)$, and 
\item $C^{\poly(\size{x})}(p \mid x) \leq (7+\delta)\cdot \log \size{x} + c$.
\end{enumerate}
\end{theorem}

\begin{acknowledge}
 The author is grateful to Marius Zimand for his observation, the Disperser~Lemma, which improved the error term in Corollary~\ref{cor: jason} from $O[\log C(x)]$ to $O(1)$ and simplified the proof of the main construction, Lemma~\ref{lem: fullk7hall}.  Thanks also to the anonymous referee for pointing out that the Expander~Lemma follows from the disperser graph of Ta-Shma, Umans, and Zuckerman (Theorem~\ref{thm: tuz}), which means that we need not appeal to the Guruswami, Umans, and Vadhan expander graph (Theorem~\ref{thm: guv09}) in the main construction.
\end{acknowledge}

\end{document}